\documentclass[copyright,creativecommons,sharealike]{eptcs}
\providecommand{\event}{ICLP 2019}
\usepackage{breakurl}

\usepackage{soul}
\usepackage{url}
\usepackage[utf8]{inputenc}
\usepackage[small]{caption}
\usepackage{graphicx}
\usepackage{amsmath}
\usepackage{amsthm}
\usepackage{booktabs}
\usepackage{algorithm}
\usepackage{algorithmic}
\usepackage[OT1]{fontenc}
\usepackage{amssymb}

\title{Epistemic Logic Programs: A Different World View}

\author{Michael Morak
  \institute{University of Klagenfurt\\Klagenfurt, Austria}
  \email{michael.morak@aau.at}
}

\def\authorrunning{Michael Morak}
\def\titlerunning{Epistemic Logic Programs: A Different World View}

\newcommand{\nop}[1]{}

\newenvironment{changemargin}[2]{%
\list{}{\rightmargin#2\leftmargin#1
\parsep=0pt\topsep=0pt\partopsep=0pt}
\item[]}
{\endlist}

\newenvironment{indented}{\begin{changemargin}{1cm}{0cm}}{\end{changemargin}}

\newtheorem{theorem}{Theorem}
\newtheorem{corollary}[theorem]{Corollary}
\newtheorem{proposition}[theorem]{Proposition}
\newtheorem{lemma}[theorem]{Lemma}
\newtheorem{definition}[theorem]{Definition}
\newtheorem{example}[theorem]{Example}

\let\phi\varphi
\let\epsilon\varepsilon

\renewcommand{\models}{\vDash}

\newcommand{\calA}{\mathcal{A}}

\newcommand{\calC}{\mathcal{C}}

\newcommand{\calE}{\mathcal{E}}

\newcommand{\calI}{\mathcal{I}}
\newcommand{\calJ}{\mathcal{J}}

\newcommand{\calL}{\mathcal{L}}
\newcommand{\calM}{\mathcal{M}}

\newcommand{\calR}{\mathcal{R}}
\newcommand{\calS}{\mathcal{S}}

\newcommand{\NP}{\ensuremath{\textsc{NP}}}
\newcommand{\co}{\ensuremath{\textsc{co}}}

\newcommand{\SIGMA}[2]{\ensuremath{\Sigma_{\mathit{#2}}^{\mathit{#1}}}}

\newcommand{\constant}[1]{\mathit{#1}}

\newcommand{\variable}[1]{\mathit{#1}}
\newcommand{\variables}[1]{{\mathbf{#1}}}

\newcommand{\mods}[1]{\mathit{mods}(#1)}
\newcommand{\answersets}[1]{\mathit{AS}(#1)}
\newcommand{\semods}[1]{\calS\calE(#1)}
\newcommand{\secwvs}[1]{\text{se-cwv}(#1)}
\newcommand{\cwvs}[1]{\text{cwv}(#1)}
\newcommand{\wvs}[1]{\text{wv}(#1)}

\newcommand{\varX}{\variable{X}}
\newcommand{\varY}{\variable{Y}}
\newcommand{\varZ}{\variable{Z}}

\newcommand{\varsX}{\variables{X}}
\newcommand{\varsY}{\variables{Y}}
\newcommand{\varsZ}{\variables{Z}}

\newcommand{\relation}[1]{{\mathit{#1}}}

\newcommand{\fullatom}[2]{{\relation{#1}(#2)}}

\newcommand{\elitof}[1]{\mathit{elit}(#1)}

\newcommand{\eneg}{\mathbf{not}\,}

\newcommand{\body}[1]{{\mathit{B}(#1)}}

\newcommand{\pbody}[1]{{\mathit{B}^+(#1)}}
\newcommand{\head}[1]{{\mathit{H}(#1)}}

\begin{document}

\maketitle

\begin{abstract}
  Epistemic Logic Programs (ELPs), an extension of Answer Set Programming (ASP)
with epistemic operators, have received renewed attention from the research
community in recent years. Classically, evaluating an ELP yields a set of world
views, with each being a set of answer sets. In this paper, we propose an
alternative definition of world views that represents them as three-valued
assignments, where each atom can be either always true, always false, or
neither. Based on several examples, we show that this definition is natural and
intuitive. We also investigate relevant computational properties of these new
semantics, and explore how other notions, like strong equivalence, are affected.

\end{abstract}

\section{Introduction}\label{sec:introduction}

Answer Set Programming (ASP) is a generic, fully declarative logic programming
language that allows users to encode problems such that the resulting output of
the program (called \emph{answer sets}) directly corresponds to solutions of the
original problem \cite{cacm:BrewkaET11,ki:SchaubW18}. Negation in ASP is
generally interpreted according to the stable model semantics
\cite{iclp:GelfondL88,ngc:GelfondL91}, that is, as negation-as-failure, also
called default negation. The negated version $\neg a$ of an atom $a$ is true if
there is no justification for $a$ in the same answer set. Hence, default
negation is a kind of ``local'' operator in the sense that it is defined
relative to the answer set that is currently considered.

Epistemic Logic Programs (ELPs) are an extension of ASP with epistemic
operators. In this paper we will mainly consider the epistemic negation operator
$\eneg$, defined in \cite{ai:ShenE16}. It allows for a form of meta-reasoning,
that is, reasoning over multiple answer sets.  Intuitively, an epistemically
negated atom $\eneg a$ expresses that $a$ is false in at least one answer set,
that is, that it is not true everywhere.  Thus, epistemic negation is defined
relative to a collection of answer sets, which \cite{ai:ShenE16} refer to as a
world view. Deciding whether such a world view exists is known to be
\SIGMA{P}{3}-complete~\cite{ai:ShenE16}, one level higher on the polynomial
hierarchy than deciding answer set existence~\cite{amai:EiterG95}.

Originally introduced as two modal operators $\mathbf{K}$ (``known'' or
``provably true'') and $\mathbf{M}$ (``possible'' or ``not provably false'') by
Gelfond \cite{aaai:Gelfond91,amai:Gelfond94}, epistemic negation in ASP has
received renewed interest (c.f.\ e.g.\
\cite{lpnmr:Gelfond11,birthday:Truszczynski11,diss:Kahl14,ijcai:CerroHS15,%
ai:ShenE16,ijcai:SonLKL17,iclp:KahlL18,aaai:FaberMW19}), with refinements of the
semantics and proposals of new language features. Further, the development of
efficient solving systems is underway with several working systems now available
\cite{logcom:KahlWBGZ15,ijcai:SonLKL17,ijcai:BichlerMW18}.

While ELPs offer the same advantages as ASP, namely a fully declarative,
intuitive language for modelling and problem encoding, in this paper we would
like to argue that the notion of the world view, as proposed in works on ELPs up
to date, is not an intuitive construct that end users of the ELP language can
directly use. Consider the following example:

\begin{example}\label{ex:intro1}
  We consider a classical example for the need of epistemic operators,
  investigated in \cite{aaai:Gelfond91,lpnmr:Gelfond11,ai:ShenE16}, namely a
  slightly simplified version of the scholarship eligibility problem. For a
  given student, we would like to decide whether they are eligible for a
  scholarship or ineligible. The rules say that someone with a high GPA is
  eligible while with a low GPA they are ineligible.
  \begin{itemize}
    \item $\relation{eligible} \gets \relation{highGPA}$
    \item $\relation{ineligible} \gets \relation{lowGPA}$
    \item $\bot \gets \relation{eligible}, \relation{ineligible}$
  \end{itemize}
  Further, we know that the student under consideration has either a high GPA or
  a fair one, but not a low one.
  \begin{itemize}
    \item $\relation{highGPA} \vee \relation{fairGPA}$
  \end{itemize}
  Finally, we use epistemic negation to say that if we can neither prove that
  the student is always eligible or always ineligible, we want to interview
  them.
  \begin{itemize}
    \item $\relation{interview} \gets \eneg \relation{eligible}, \eneg
      \relation{ineligible}$
  \end{itemize}
  According to \cite[Example~3]{ai:ShenE16} there is exactly one world view,
  namely the set of answer sets $\{ M_1, M_2 \}$ where the answer set $M_1 = \{
    \relation{fairGPA}, \relation{interview} \}$ and answer set $M_2 = \{
      \relation{highGPA}, \relation{eligible}, \relation{interview} \}$.
\end{example}

In the above example, it is not easy to extract the relevant information for the
end user from the world view. In fact, they would have to check each answer set
to answer their original question about whether students are eligible,
ineligible, or should be interviewed. In the example above,
$\relation{interview}$ is true in all answer sets of the world view, and hence,
the student should be interviewed. In addition, the individual answer sets
contain a lot of information not directly relevant to the question of the user.
A result that directly represents the desired results, that is, a different
version of the world view that simply states that $\relation{interview}$ is
always true, while $\relation{ineligible}$ is always false, would, in the
authors' opinion, be a more intuitive representation of the result of the ELP of
Example~\ref{ex:intro1}, as it would directly answer the question whether the
student under consideration is eligible, ineligible, or should be interviewed.
The aim of this paper is to propose such a notion, and to investigate its
consequences for evaluating ELPs.

\paragraph{Contributions.} The results and contributions presented in this paper
can be summarized as follows.

\begin{itemize}
  \item We provide a different definition of the notion of ``world view'' for
    ELPs, such that relevant information about the solution to an ELP is
    represented directly, and hence easily accessible to an end user. This is
    done by defining the world view not as a collection of answer sets, but as a
    three-valued model that encodes this information.
  \item We investigate the relationship of our new world view definition with
    the semantics of ELPs provided by Shen and Eiter \cite{ai:ShenE16} and
    establish a close connection between the notion of the epistemic guess
    proposed by these authors, and our new definition of a world view.
  \item We investigate notions of equivalence under our new semantics. Here, we
    show that equivalence under our definition of a world view is more general
    than the one under the world view definition of \cite{ai:ShenE16}, but that
    for strong equivalence, the two semantics coincide.
  \item We study the computational complexity of the world view existence
    problem, and of equivalence testing.
  \item Finally, as a case study, we investigate the problem of QSAT-solving,
    and show that our notion of the world view can more intuitively capture the
    relevant part of the solution to a given QSAT problem.
\end{itemize}

\paragraph{Structure.} The remainder of the paper is structured as follows. In
Section~\ref{sec:preliminaries}, we provide an overview of ASP and ELPs. In
Section~\ref{sec:worldviews}, our novel definition of the world view is proposed
and investigated. Notions of equivalence are studied in
Section~\ref{sec:equivalence}, and a case study is performed in
Section~\ref{sec:casestudies}. We conclude with a summary in
Section~\ref{sec:conclusions}.

\section{Preliminaries}\label{sec:preliminaries}

This section provides relevant details on Answer Set Programming and Epistemic
Logic Programs.

\subsection{Answer Set Programming (ASP)}

A \emph{ground logic program} with double negation (also called answer set
program, ASP program, or, simply, logic program) is a pair $\Pi = (\calA,
\calR)$, where $\calA$ is a set of propositional (i.e.\ ground) atoms and
$\calR$ is a set of rules of the form
\begin{equation}\label{eq:rule}
  a_1\vee \cdots \vee a_l \leftarrow a_{l+1}, \ldots, a_m, \neg \ell_1, \ldots,
  \neg \ell_n;
\end{equation}
where the comma symbol stands for conjunction, $0 \leq l \leq m$, $0 \leq n$,
$a_i \in \calA$ for all $1 \leq i \leq m$, and each $\ell_i$ is a
\emph{literal}, that is, either an atom $a$ or its (default) negation $\neg a$
for any atom $a \in \calA$. Note that, therefore, doubly negated atoms may
occur. Each rule $r \in \calR$ of form~(\ref{eq:rule}) consists of a \emph{head}
$\head{r} = \{ a_1,\ldots,a_l \}$ and a \emph{body} $\body{r} =
\{a_{l+1},\ldots,a_m, \neg \ell_1, \ldots, \neg \ell_n \}$. We denote the
\emph{positive} body by $\pbody{r} = \{ a_{l+1}, \ldots, a_m \}$.

An \emph{interpretation} $I$ (over $\calA$) is a set of atoms, that is, $I
\subseteq \calA$.  A literal $\ell$ is true in an interpretation $I \subseteq
\calA$, denoted $I \models \ell$, if $a \in I$ and $\ell = a$, or if $a \not\in
I$ and $\ell = \neg a$; otherwise $\ell$ is false in $I$, denoted $I \not\models
\ell$. Finally, for some literal $\ell$, we define that $I \models \neg \ell$ if
$I \not\models \ell$. This notation naturally extends to sets of literals. An
interpretation $M$ is called a \emph{model} of $r$, denoted $M \models r$, if,
whenever $M \models \body{r}$, it holds that $M \models \head{r}$. We denote the
set of models of $r$ by $\mods{r}$; the models of a logic program $\Pi=
(\calA,\calR)$ are given by $\mods{\Pi} = \bigcap_{r \in \calR} \mods{r}$. We
also write $I\models r$ (resp.\ $I\models \Pi$) if $I\in\mods{r}$ (resp.\
$I\in\mods{\Pi}$).

The GL-reduct $\Pi^I$ of a logic program $\Pi = (\calA, \calR)$ w.r.t.\ an
interpretation $I$ is defined as $\Pi^I = (\calA, \calR^I)$, where $\calR^I = \{
  \head{r} \leftarrow \pbody{r} \mid r \in \calR, \forall \neg \ell \in \body{r}
: I \models \neg \ell \}$.

\begin{definition}\label{def:answerset}
  \cite{iclp:GelfondL88,ngc:GelfondL91,amai:LifschitzTT99} $M \subseteq \calA$
  is an \emph{answer set} of a logic program $\Pi$ if (1) $M \in \mods{\Pi}$ and
  (2) there is no subset $M' \subset M$ such that $M' \in \mods{\Pi^M}$.
\end{definition}

The set of answer sets of a logic program $\Pi$ is denoted by
$\answersets{\Pi}$.  The \emph{consistency problem} of ASP, that is, to decide
whether for a given logic program $\Pi$ it holds that
$\answersets{\Pi}\neq\emptyset$, is \SIGMA{P}{2}-complete~\cite{amai:EiterG95},
and remains so also in the case where doubly negated atoms are allowed in rule
bodies~\cite{tplp:PearceTW09}.

\paragraph{Strong Equivalence for Logic Programs.} Two logic programs $\Pi_1 =
(\calA, \calR_1)$ and $\Pi_2 = (\calA, \calR_2)$ are \emph{equivalent} iff they
have the same set of answer sets, that is, $\answersets{\Pi_1} =
\answersets{\Pi_2}$. The two logic programs are \emph{strongly equivalent} iff
for any third logic program $\Pi = (\calA, \calR)$ it holds that the combined
logic program $\Pi_1 \cup \Pi = (\calA, \calR_1 \cup \calR)$ is equivalent to
the combined logic program $\Pi_2 \cup \Pi = (\calA, \calR_2 \cup \calR)$. Note
that strong equivalence implies equivalence, since the empty program
$\Pi=(\calA,\emptyset)$ would already contradict strong equivalence for two
non-equivalent programs $\Pi_1$ and $\Pi_2$.

An \emph{SE-model} \cite{tplp:Turner03} of a logic program $\Pi = (\calA,
\calR)$ is a two-tuple of interpretations $(X, Y)$, where $X \subseteq
Y\subseteq\calA$, $Y \models \Pi$, and $X \models \Pi^Y$. The set of SE-models
of a logic program $\Pi$ is denoted $\semods{\Pi}$. Note that for every model
$Y$ of $\Pi$, $(Y, Y)$ is an SE-model of $\Pi$, since $Y \models \Pi$ implies
that $Y \models \Pi^Y$.

Two logic programs (over the same atoms) are strongly equivalent iff they have
the same SE-models \cite{tocl:LifschitzPV01,tplp:Turner03}. Checking whether two
logic programs are strongly equivalent is known to be \co\NP-complete
\cite{tplp:Turner03,tplp:PearceTW09}.

\subsection{Epistemic Logic Programs (ELPs)}

An \emph{epistemic literal} is a formula $\eneg \ell$, where $\ell$ is a literal
and $\eneg$ is the epistemic negation operator. A \emph{ground epistemic logic
program (ELP)} is a tuple $\Pi = (\calA, \calR)$, where $\calA$ is a set of
propositional atoms and $\calR$ is a set of rules of the following form:
\begin{equation*}
   a_1\vee \cdots \vee a_k \leftarrow \ell_1, \ldots, \ell_m, \xi_1, \ldots,
   \xi_j, \neg \xi_{j + 1}, \ldots, \neg \xi_{n},
\end{equation*}
where $k \geqslant 0$, $m \geqslant 0$, $n \geqslant j \geqslant 0$, each $a_i
\in \calA$ is an atom, each $\ell_i$ is a literal, and each $\xi_i$ is an
epistemic literal, where the latter two each use an atom from $\calA$. Such
rules are also called \emph{ELP rules}.

Similar to logic programs, let $\head{r} = \{ a_1, \ldots, a_k \}$ denote the
head elements of an ELP rule, and let $\body{r} = \{ \ell_1, \ldots, \ell_m,
\xi_1, \ldots, \xi_j, \neg \xi_{j+1}, \ldots, \neg \xi_{n} \}$, that is, the set
of elements appearing in the rule body. The \emph{union} (or \emph{combination})
of two ELPs $\Pi_1 = (\calA_1, \calR_1)$ and $\Pi_2 = (\calA_2, \calR_2)$ is the
ELP $\Pi_1 \cup \Pi_2 = (\calA_1 \cup \calA_2, \calR_1 \cup \calR_2)$. The set
of epistemic literals occurring in an ELP $\Pi$ is denoted $\elitof{\Pi}$.

Shen and Eiter \cite{ai:ShenE16} define the semantics of ELPs as follows: let
$\Pi = (\calA, \calR)$ be an ELP. A subset $\Phi \subseteq \elitof{\Pi}$ is
called an \emph{epistemic guess} (or, simply, a \emph{guess}).  Given such a
guess, the \emph{epistemic reduct} of $\Pi$ w.r.t.\ $\Phi$, denoted $\Pi^\Phi$,
consists of the rules $\{ r^\neg \mid r \in \calR \}$, where $r^\neg$ is defined
as the rule $r \in \calR$ where all occurrences of epistemic literals $\eneg
\ell \in \Phi$ are replaced by $\top$, and all remaining epistemic negation
symbols $\eneg$ are replaced by default negation $\neg$. Note that, after this
transformation, $\Pi^\Phi$ is a logic program without epistemic
negation\footnote{Note that $\Pi^\Phi$ may contain triple-negated atoms
$\neg\neg\neg a$. But such formulas can always be replaced by $\neg a$
\cite{amai:LifschitzTT99}.}.

We are now ready to give the following, central definition:

\begin{definition}\label{def:secandidateworldview}
  Let $\Pi = (\calA, \calR)$ be an ELP. A set $\calM$ of interpretations over
  $\calA$ is a \emph{Shen-Eiter candidate world view (SE-CWV)} of $\Pi$ if there
  is an epistemic guess $\Phi \subseteq \elitof{\Pi}$ such that $\calM =
  \answersets{\Pi^\Phi}$ and the following conditions hold:
  \begin{enumerate}
    \item\label{def:secompatibility:1} $\calM \neq \emptyset$;
    \item\label{def:secompatibility:2} for each epistemic literal $\eneg \ell
      \in \Phi$, there exists an answer set $M \in \calM$ such that $M
      \not\models \ell$; and
    \item\label{def:secompatibility:3} for each epistemic literal $\eneg \ell
      \in \elitof{\Pi} \setminus \Phi$, for all answer sets $M \in \calM$ it
      holds that $M \models \ell$.
  \end{enumerate}
\end{definition}

The set of all SE-CWVs of an ELP $\Pi$ is denoted by $\secwvs{\Pi}$.  Following
the principle of knowledge minimization, \cite{ai:ShenE16} define a
\emph{Shen-Eiter world view} as a SE-CWV that has a maximal guess.

\begin{definition}\label{def:seworldview}
  Let $\Pi = (\calA, \calR)$ be an ELP. An SE-CWV $\calC$ is called a
  \emph{Shen-Eiter world view (SE-WV)} if its associated guess $\Phi$ is
  subset-maximal w.r.t.\ $\elitof{\Pi}$.
\end{definition}

The main reasoning task for ELPs is checking whether they are \emph{consistent},
that is, whether it holds that $\secwvs{\Pi} \neq \emptyset$. This problem is
also referred to as \emph{world view existence
problem}\footnote{SE-CWV-existence and SE-WV-existence are equivalent.}. This
problem is known to be \SIGMA{P}{3}-complete \cite{ai:ShenE16}.

\section{A Different World View}\label{sec:worldviews}

In this section, we will introduce our novel world view definition and explore
the relationship with the existing world view definition of Shen and Eiter
\cite{ai:ShenE16}, as defined in Section \ref{sec:preliminaries}. We also look
at the computational complexity of deciding whether a world view exists.

\subsection{The World View as a Model}

The idea underlying our definition is to treat world views not as collections of
answer sets, which are hard to interpret from an intuitive perspective, but,
instead, as three-valued models of ELPs. We first define the notion of a world
interpretation, which will be in a similar relationship to our new world view,
as an interpretation is to a model in classical logic.

\begin{definition}\label{def:cwi}
  Let $\Pi = (\calA, \calR)$ be an ELP, and let $\calL$ be the set of all
  literals built from atoms in $\calA$. Then, we call a consistent subset $I
  \subseteq \calL$ a \emph{candidate world interpretation (CWI)} for $\Pi$.
\end{definition}

This definition gives rise to a total, three-valued truth assignment to the
atoms $\calA$ of an ELP $\Pi = (\calA, \calR)$; hence, we will sometimes treat a
CWI $I$ as a triple of disjoint sets $(I^P, I^N, I^U)$, where $I^P = \{ a \mid a
\in I \}$, $I^N = \{ a \mid \neg a \in I \}$ and $I^U = (\calA \setminus I^P)
\setminus I^N$.

Note the immediate intuitive meaning of this three-valued assignment: an atom
$a$ in $I^P$ (i.e.\ $a \in I$) means that it is \emph{always true}, in $I^N$
(i.e.\ $\neg a \in I$) it means it is \emph{always false}, and in $I^U$ it is
\emph{unknown} in the sense that it is neither always false nor always true.
Note further that CWIs naturally represent consistent epistemic guesses in the
semantics of Shen and Eiter: assume an ELP $\Pi = (\calA, \calR)$ contains
exactly two epistemic literals, $\eneg a$ and $\eneg \neg a$. A-priori, from the
definitions in \cite{ai:ShenE16}, one may think that there are four possible
epistemic guesses: $\Phi_1 = \{ \eneg a, \eneg \neg a \}$, $\Phi_2 = \{ \eneg a
\}$, $\Phi_3 = \{ \eneg \neg a \}$, and $\Phi_4 = \emptyset$. But it turns out
that $\Phi_4$ will never give rise to a SE-CWV: we cannot both believe that
$\eneg a$ is false and that also $\eneg \neg a$ is false, because that would
mean that in the epistemic reduct $\Pi^{\Phi_4}$, $a$ has to be true in every
answer set of the epistemic reduct, but also that $\neg a$ has to be true in
every such answer set. This is only possible if $\Pi^{\Phi_4}$ has no answer
set, but then $\Phi_4$ can, by definition, not lead to a SE-CWV. Hence $\Phi_4$
is an inconsistent guess, which is naturally excluded by our definition of CWIs.

The other three epistemic guesses, $\Phi_1$, $\Phi_2$, and $\Phi_3$, however,
correspond to three distinct truth assignments to $a$. $\Phi_2$ represents that
we believe $\eneg a$ but not $\eneg \neg a$. The former says that there must be
an answer set of $\Pi^{\Phi_2}$ where $a$ is false, and the latter says that in
all answer sets of $\Pi^{\Phi_2}$, $\neg a$ must be true. Hence, $\Phi_2$
represents the epistemic guess that says that $a$ must be false everywhere.
Similarly, $\Phi_3$ says that $a$ must be true everywhere. Finally, $\Phi_1$
says that there must be an answer set of $\Pi^{\Phi_1}$ where $a$ is true, and
one where $a$ is false, and hence, that $a$ shall neither be true everywhere, or
false everywhere. Since the intuitive meaning of these three guesses can only be
discerned by referring back to the set of epistemic literals that appear in
$\Pi$. In fact, for ELPs where $\eneg \neg a$ does not appear, the meaning of
guess $\Phi_2$ would be different: instead of saying that $a$ must be false
everywhere, it would then say that $a$ must be false somewhere (but not
necessarily everywhere). It is thus our contention that epistemic guesses, as
defined in \cite{ai:ShenE16}, represent their intuitive meaning in a rather
convoluted and hard-to-understand way. Now note that the three guesses $\Phi_1$,
$\Phi_2$, and $\Phi_3$ naturally correspond to the intuitive meaning of our
CWIs: assuming that the universe of atoms $\calA$ contains only the atom $a$,
$\Phi_1$ corresponds to the CWI $\emptyset$, $\Phi_2$ to the CWI $\{ \neg a \}$
and $\Phi_3$ to the CWI $\{ a \}$. Thus, in the authors' opinion, CWIs capture
this intuitive meaning much better than the notion of the epistemic guess.

Before we can give our novel definition of a world view, we define the notion of
the epistemic reduct w.r.t.\ a CWI, in the same style as in \cite{ai:ShenE16},
as follows:

\begin{definition}\label{def:epistemicreduct}
  Let $\Pi = (\calA, \calR)$ be an ELP, and let $I$ be a CWI for $\Pi$. Then,
  let $\Pi^I$ denote the \emph{epistemic reduct} of $\Pi$ w.r.t.\ $I$, defined
  as follows: $\Pi^I = (\calA, \calR')$, where $\calR' = \{ r^I \mid r \in \calR
  \}$ and $r^I$ denotes the rule $r$ where each epistemic literal $\eneg \ell$
  is replaced by $\neg \ell$ if $\ell \in I$, and by $\top$ otherwise.
\end{definition}

We now define when a set of interpretations is compatible with a CWI.

\begin{definition}\label{def:compatibility}
  Let $\calI$ be a set of interpretations over a set of atoms $\calA$. Then, a
  CWI $I$ is \emph{compatible} with $\calI$ iff the following conditions hold:
  \begin{enumerate}
    \item\label{def:compatibility:1} $\calI \neq \emptyset$;
    \item\label{def:compatibility:2} for each atom $a \in I^P$, it holds that
      for each interpretation $J \in \calI$, $a \in J$;
    \item\label{def:compatibility:3} for each atom $a \in I^N$, it holds that
      for each interpretation $J \in \calI$, $a \not\in J$;
    \item\label{def:compatibility:4} for atom $a \in I^U$, it holds that there
      are two interpretations $J, J' \in \calI$, such that $a \in J$, but $a
      \not\in J'$.
  \end{enumerate}
\end{definition}

With these definitions in place, we can now give the main definition of this
work, our novel definition of a world view:

\begin{definition}\label{def:worldview}
  Let $\Pi = (\calA, \calR)$ be an ELP, and let $I$ be a CWI for $\Pi$. Then,
  $I$ is a \emph{candidate world view (CWV)} of $\Pi$ iff the set of answer sets
  $\answersets{\Pi^I}$ is compatible with $I$. A CWV $I$ is called a \emph{world
  view (WV)} if it is subset-minimal, that is, there is no other CWV $J \subset
  I$ of $\Pi$. The set of CWVs (resp.\ WVs) of an ELP $\Pi$ is denoted
  $\cwvs{\Pi}$ (resp.\ $\wvs{\Pi}$).
\end{definition}

To illustrate the usefulness of our newly proposed notion of WVs, let us return
to our introductory example from Section~\ref{sec:introduction}.

\begin{example}\label{ex:intro2}
  Consider the ELP $\Pi = (\calA, \calR)$, where the set of atoms $\calA$ is the
  set of all the atoms appearing in Example~\ref{ex:intro1}, that is, $\calA =
  \{ \relation{eligible}, \relation{ineligible}, \relation{highGPA},
  \relation{lowGPA}, \relation{fairGPA}, \relation{interview} \}$, and where
  $\calR$ contains the rules from Example~\ref{ex:intro1}. As stated in
  Section~\ref{sec:introduction}, this ELP has precisely one SE-CWV (and hence
  SE-WV). This SE-CWV corresponds to the epistemic guess $\Phi$, where $\Phi =
  \{ \eneg \relation{eligible}, \eneg \relation{ineligible} \}$. As stated in
  Example~\ref{ex:intro1}, both answer sets in this SE-CWV contain the fact
  $\relation{interview}$, and hence the candidate should be interviewed.

  Using our new notion of CWVs, we again see that there is exactly one CWV (and
  hence WV) $W$, where we have that $W = \{ \relation{interview}, \neg
  \relation{ineligible} \}$; or, equivalently, where $W^P = \{
  \relation{interview} \}$, $W^N = \{ \relation{ineligible} \}$, and $W^U =
  (\calA \setminus W^P) \setminus W^N$.  From our CWV $W$, we can immediately
  recognize the solution to the scholarship eligibility problem; namely, in this
  case, to interview the student (since the fact $\relation{interview}$ is true
  in $W$, and hence, intuitively, ``true in every possible world''). For
  SE-CWVs, the individual answer sets first need to be examined and compared to
  each other to draw that conclusion.
\end{example}

Emphasized by the above example, we would argue that our novel definition of a
world view is much more intuitive than the standard definition as a set of
answer sets. This is mainly because it has the distinct advantage of being
immediately interpretable: it is simply a truth assignment to all the atoms in
the ELP, with the additional information that some atoms are both true and false
somewhere (i.e.\ unknown). This information is suitable to be used directly by
end-users of ELP solvers.

With our main definitions in place, next we will investigate the relationship
between SE-CWVs and our novel CWVs as defined above.

\subsection{Relationship to SE-CWVs}

In this subsection, we will show that there is a close relationship between
SE-CWVs and CWVs. As is obvious from the definitions, the conditions that define
these notions are similar. The main difference is that our CWVs correspond to
the epistemic guesses in \cite{ai:ShenE16} instead of to the actual SE-CWVs,
those being sets of answer sets. However, for an ELP $\Pi = (\calA, \calR)$,
note that our CWVs are actually \emph{total} three-valued assignments to the
atoms in $\calA$; that is, each atom is either always true, always false, or
neither (i.e.\ unknown). Epistemic guesses, on the other hand, operate only on
those epistemic literals that actually appear in $\Pi$. Nevertheless, we can
show that epistemic guesses and our CWVs are closely related. Firstly, note that
we can always modify an ELP such that all possible epistemic literals actually
appear.

\begin{lemma}\label{lem:elitdomain}
  Let $\Pi = (\calA, \calR)$ be an ELP and let $\Pi' = (\calA, \calR')$ be an
  ELP over the same atom domain, but with $\calR' \supset \calR$ and $\calR'
  \setminus \calR = \{ r_1 \}$, where $r_1$ is the rule $\bot \gets a, \neg a,
  \eneg \ell$, where $\eneg \ell$ does not appear in $\calR$ and $\ell$ is a
  literal over the atom $a \in \calA$. Then, $\secwvs{\Pi} = \secwvs{\Pi'}$.
\end{lemma}

\begin{proof}
  Note that the literal $\ell$ must be over some atom from $\calA$. Note further
  that the rule $r_1$ is clearly tautological, that is, it will always be
  satisfied. Now, consider any epistemic guess $\Phi \subseteq \elitof{\Pi}$,
  and let $\Phi' = \Phi \cup \{ \eneg \ell \}$. Observe that $\Pi'^{\Phi'} =
  \Pi'^\Phi$, since, in the epistemic reduct, rule $r_1$ reduces to $\bot \gets
  a, \neg a$ for both guesses $\Phi$ and $\Phi'$. Hence, let $\calM =
  \answersets{\Pi'^\Phi}$. Since $r_1$ is tautological, it has no influence on
  the answer sets, and we conclude that $\calM = \answersets{\Pi^\Phi}$.

  Now, assume that $\calM$ is an SE-CWV for $\Pi$ w.r.t.\ the epistemic guess
  $\Phi$. Then, it is not difficult to see that it also satisfies the conditions
  of Definition~\ref{def:secandidateworldview} for $\Pi'$ w.r.t.\ either the
  epistemic guess $\Phi$ or $\Phi'$ (since either $\ell$ is true in all answer
  sets in $\calM$, or there is one answer set where $\ell$ is false). Hence,
  $\calM$ is also an SE-CWV for $\Pi'$, either w.r.t.\ guess $\Phi$ or $\Phi'$.

  Conversely, assume that $\calM$ is not an SE-CWV for $\Pi$ w.r.t.\ guess
  $\Phi$. Then, either $\calM = \emptyset$, or there is an epistemic literal
  $\eneg \ell' \in \Phi$ that violates Condition~\ref{def:secompatibility:2}
  or~\ref{def:secompatibility:3} of Definition~\ref{def:secandidateworldview}.
  But then, this violation also holds for $\calM$ w.r.t.\ the ELP $\Pi'$ and
  both of the guesses $\Phi$ and $\Phi'$. Hence, $\calM$ is also not an SE-CWV
  for $\Pi'$.
\end{proof}

We are now ready to state our main correspondence theorem between our notion of
world views and the notion by Shen and Eiter \cite{ai:ShenE16}. In fact, we show
that there is a one-to-one correspondence between SE-CWVs and our CWVs.

\begin{theorem}\label{thm:cvwcorrespondence}
  Let $\Pi = (\calA, \calR)$ be an ELP. Then, it holds that for each CWV of
  $\Pi$ there is exactly one SE-CWV of $\Pi$, and for each SE-CWV of $\Pi$ there
  is exactly one CWV of $\Pi$.
\end{theorem}

\begin{proof}[Proof (Sketch)]
  Firstly, note that by repeated applications of Lemma~\ref{lem:elitdomain}, we
  can always assume that the rules of $\Pi$ contain all the epistemic literals
  that can be build from $\calA$.
  
  To prove the first direction, assume that $W$ is a CWV of $\Pi$. Let $\Phi$ be
  the following epistemic guess for $\Pi$: for each $a \in W^P$, let $\Phi \cap
  \{ \eneg a, \eneg \neg a \} = \{ \eneg \neg a \}$; for each $a \in W^N$, let
  $\Phi \cap \{ \eneg a, \eneg \neg a \} = \{ \eneg a \}$; and for each $a \in
  W^U$, let $\Phi \cap \{ \eneg a, \eneg \neg a \} = \{ \eneg a, \eneg \neg a
  \}$. It can be verified that the epistemic reduct $\Pi^W$ is equal to the
  Shen-Eiter epistemic reduct $\Pi^\Phi$. Hence, let $\calM = \answersets{\Pi^W}
  = \answersets{\Pi^\Phi}$. Since $W$ is, by assumption, a CWV of $\Pi$, $\calM$
  fulfills the conditions of Definition~\ref{def:compatibility} w.r.t.\ $W$.
  But then it is not difficult to see that, by construction of $\Phi$, $\calM$
  also satisfies the conditions of Definition~\ref{def:secandidateworldview},
  and hence is an SE-CWV of $\Pi$ w.r.t.\ guess $\Phi$.

  For the other direction, assume that $\calM$ is an SE-CWV for $\Pi$ w.r.t.\
  some epistemic guess $\Phi$. Now, construct CWI $W$ as follows (recall that
  $\Phi$ is a subset of all possible epistemic literals over $\calA$): for each
  $a \in \calA$, if $\Phi \cap \{ \eneg a, \eneg \neg a \} = \{ \eneg \neg a \}$
  then let $a \in W^P$; if $\Phi \cap \{ \eneg a, \eneg \neg a \} = \{ \eneg a
  \}$ then let $a \in W^N$; and if $\Phi \cap \{ \eneg a, \eneg \neg a \} = \{
  \eneg a, \eneg \neg a \}$ then let $a \in W^U$. By a similar argument to the
  one above, we can show that $W$ is indeed a CWV of $\Pi$, since, by
  construction of $W$, the epistemic reducts $\Pi^\Phi$ and $\Pi^W$ coincide.
\end{proof}

The proof for the theorem above gives a reduction between SE-CWV-existence and
CWV-existence and vice-versa. From this reduction, and the fact that
SE-CWV-existence is \SIGMA{P}{3}-complete in general \cite{ai:ShenE16}, the
statement below follows immediately:

\begin{theorem}
  Checking whether an ELP has at least one CWV (or, equivalently, at least one WV) is \SIGMA{P}{3}-complete.
\end{theorem}

The statement above also holds for WVs, because if a CWV exists then, clearly,
also a subset-minimal CWV (that is, a WV) exists. Furthermore, any WV is also a
CWV. This concludes our investigation of the novel world view definition
proposed in this section.

\section{Equivalence of ELPs}\label{sec:equivalence}

In this section, we will look at different equivalence notions for ELPs. But
before we start, let us first explore several properties of ELPs that we will
make use of in this section.

We start by re-stating a folklore result from the world of ASP, namely that the
universe of atoms of an answer set program can be extended without changing its
answer sets.

\begin{proposition}\label{prop:atomdomainlp}
  Let $\Pi = (\calA, \calR)$ be a logic program and let $\Pi' = (\calA', \calR)$
  be a logic program with the same set of rules, but with $\calA' \supset
  \calA$. Then, $\answersets{\Pi} = \answersets{\Pi'}$.
\end{proposition}

It is easy to see that the above proposition holds by noting that any atom $a
\in \calA' \setminus \calA$ can clearly not appear in the rules $\calR$, and
hence any such $a$ must be false in any answer set of $\Pi'$ (i.e.\ for all $M
\in \answersets{\Pi'}$ it holds that $a \not\in M$). From this result, it is
easy to obtain a similar result for ELPs:

\begin{proposition}\label{prop:atomdomainelp}
  Let $\Pi = (\calA, \calR)$ be an ELP and let $\Pi' = (\calA', \calR)$ be an
  ELP with the same set of rules, but with $\calA' \supset \calA$. Then,
  $\cwvs{\Pi'} = \{ W \cup \{ \neg a \} \mid W \in \cwvs{\Pi} \}$.
\end{proposition}

\begin{proof}
  Note that, clearly, any atom $a \in \calA' \setminus \calA$ cannot appear
  anywhere in $\calR$. Hence, for every CWI $I$ over $\calA$, $a$ also does not
  appear in the rules of the epistemic reduct $\Pi'^I$ and thus also not in any
  answer set of $\Pi'^I$. But then, $\answersets{\Pi'^I} = \answersets{\Pi^I}$
  (via Proposition~\ref{prop:atomdomainlp}). Then, $a$ is false in every one of
  these answer sets. It is therefore trivial that if $W$ is a CWV of $\Pi$, then
  $W \cup \{ \neg a \}$ is a CWV of $\Pi'$.
\end{proof}

From the above, we can see that the atom domain of an ELP can be arbitrarily
extended, and atoms that are in the atom domain of an ELP but do not appear in
its rules will always be false in every CWV of the ELP. With these results we
are now ready to investigate notions of equivalence for ELPs under the semantics
proposed in Section~\ref{sec:worldviews}.

In a recent paper that investigates strong equivalence for ELPs
\cite{aaai:FaberMW19}, the authors define equivalence of ELPs under the SE-CWV
semantics as follows: two ELPs are \emph{SE-CWV-equivalent} iff their SE-CWVs
are the same and \emph{SE-WV-equivalent} iff their SE-WVs are the same. We can
define a similar notion of equivalence using our definition of CWVs as follows:

\begin{definition}\label{def:equivalence}
  Two ELPs $\Pi_1$ and $\Pi_2$ are \emph{CWV-equivalent} (resp.\
  \emph{WV-equivalent}), denoted $\Pi_1 \equiv_{CWV} \Pi_2$ (resp.\ $\Pi_1
  \equiv_{WV} \Pi_2$) if and only if $\cwvs{\Pi_1} = \cwvs{\Pi_2}$ (resp.\
  $\wvs{\Pi_1} = \wvs{\Pi_2})$.
\end{definition}

Recall that our new notion of CWVs abstracts from the answer sets, since the
individual answer sets themselves no longer appear directly within them. Hence,
our notion of equivalence turns out to be strictly more general than the one in
\cite{aaai:FaberMW19}: we can find an ELP that is equivalent under our
semantics, but not equivalent under their notion of equivalence, which is based
on SE-CWVs, as the following theorem states.

\begin{theorem}\label{thm:equivalence}
  CWV-equivalence (resp.\ WV-equivalence) strictly generalizes
  SE-CWV-equivalence (resp.\ SE-WV-equivalence).
\end{theorem}

\begin{proof}
  Clearly, SE-(C)WV-equivalence implies (C)WV-equivalence, since when the
  SE-(C)WVs are the same, then so are the (C)WVs. To see the other direction, we
  can construct two ELPs $\Pi_1$ and $\Pi_2$, such that they both have exactly
  one CWV (and hence WV, SE-CWV, and SE-WV), and that the SE-CWV of $\Pi_1$ is
  $\{ \{ a \}, \{ b \}, \{ c \} \}$. Hence, the corresponding CWV $W$ assigns
  all three atoms $a$, $b$, and $c$ to unknown (i.e.\ $W^U = \{ a, b, c \}$).
  Now, we can construct an ELP $\Pi_2$ in such a way that it has the CWV $W$
  (that is, all three atoms are still unknown), but the corresponding SE-CWV of
  $\Pi_2$ is $\{ \{ a \}, \{ b, c \} \}$. This shows that $\Pi_1$ and $\Pi_2$
  are (C)WV-equivalent, but not SE-(C)WV-equivalent, as desired.
\end{proof}

With the definition for equivalence in place, we can now also investigate the
notion of strong equivalence. Strong equivalence is a well-studied problem in
plain ASP, with useful applications for program simplification
\cite{tocl:LifschitzPV01,iclp:CabalarPV07,jair:LinC07,jancl:EiterFPTW13}. For
ELPs, several works deal with strong equivalence by abstracting into epistemic
extensions of Heyting's logic of here-and-there (HT)
\cite{lpnmr:WangZ05,ijcai:CerroHS15}. A recent paper deals with a direct
characterization of strong equivalence for the SE-CWV semantics
\cite{aaai:FaberMW19}. Similarly to the definition of ordinary equivalence, we
will see how we can generalize their characterization to apply to our semantics.
We begin by defining strong equivalence in our context:

\begin{definition}\label{def:strongequivalence}
  Two ELPs $\Pi_1$ and $\Pi_2$ are \emph{strongly CWV-equivalent} (resp.\
  \emph{strongly WV-equivalent}) if and only if for any third ELP $\Pi$, it
  holds that $\Pi_1 \cup \Pi \equiv_{CWV} \Pi_2 \cup \Pi$ (resp.\ $\Pi_1 \cup
  \Pi \equiv_{WV} \Pi_2 \cup \Pi$).
\end{definition}

Along the lines of \cite{aaai:FaberMW19}, we characterize the notion of strong
equivalence via a so-called SE-function, defined as follows:

\begin{definition}\label{def:sefunction}
  The \emph{SE-function} $\calS\calE_\Pi(\cdot)$ of an ELP $\Pi = (\calA,
  \calR)$  maps CWIs $I$ over $\calA$ to sets of SE-models as follows.
  $$\calS\calE_\Pi(I) = \left\{
  \begin{array}{@{}lr@{}}
    \semods{\Pi^I} \qquad & \text{if } I \text{ is compatible with some}\\
    & \text{subset } \calJ \subseteq \mods{\Pi^I}\\[1.5ex]
    \emptyset & \text{otherwise.}
  \end{array}
  \right.$$
\end{definition}

It turns out that this modified version of the SE-function precisely
characterizes strong equivalence in our setting, and that the two notions of
strong equivalence given in Definition~\ref{def:strongequivalence} coincide, as
the following theorem states:

\begin{theorem}\label{thm:strongequivalence}
  Let $\Pi_1$ and $\Pi_2$ be two ELPs. Then, the following statements are
  equivalent:
  \begin{enumerate}
    \item $\Pi_1$ and $\Pi_2$ are strongly CWV-equivalent;
    \item $\Pi_1$ and $\Pi_2$ are strongly WV-equivalent; and
    \item $\calS\calE_{\Pi_1} = \calS\calE_{\Pi_2}$.
  \end{enumerate} 
\end{theorem}

\begin{proof}[Proof (Sketch)]
  Via Proposition~\ref{prop:atomdomainelp}, we can always assume that $\Pi_1$
  and $\Pi_2$ have the same atom domain. Now, $(1) \Rightarrow (2)$ follows from
  Definition~\ref{def:strongequivalence} and the fact that every WV is a CWV.
  $(3) \Rightarrow (1)$ follows by the same argument as $(5) \Rightarrow (1)$ in
  the proof of Theorem~15 in \cite{aaai:FaberMW19}. Finally, $(2) \Rightarrow
  (3)$ can be established by showing the contrapositive. We can construct $\Pi$
  in such a way that it realizes a particular CWI in the SE-function as a CWV
  for both $\Pi_1 \cup \Pi$ and $\Pi_2 \cup \Pi$ (which, by construction of the
  SE-function, is always possible). We chose a CWI $I$ where the SE-function
  contains a difference. We then make use of the fact that, for this CWI $I$,
  there is a difference in the SE-models of the epistemic reducts $(\Pi_1 \cup
  \Pi)^I$ and $(\Pi_2 \cup \Pi)^I$. However, simply realizing this difference is
  not enough to show non-strong-equivalence (as opposed to the SE-CWV semantics,
  where this would already suffice).  Instead, we need to make sure that $I$
  actually becomes a CWV only in $\Pi_1 \cup \Pi$ but not for $\Pi_2 \cup \Pi$
  (w.l.o.g., by symmetry). This can be done by introducing a new atom, say $a$,
  not appearing in $\Pi_1$ or $\Pi_2$, and constructing $\Pi$ in such a way that
  this atom is true precisely in the answer set created from the SE-model that
  marks the difference in the SE-functions of $\Pi_1$ and $\Pi_2$. Thus, since
  this atom $a$ is false in every answer set of $(\Pi_1 \cup \Pi)^I$, but true
  in exactly one answer set of $(\Pi_2 \cup \Pi)^I$, clearly, $I$ cannot be a
  CWV of both, showing $(2) \Rightarrow (3)$ via the contrapositive.
\end{proof}

This shows that the two notions of strong WV-equivalence and strong
CWV-equivalence coincide. Interestingly, since our construction of the
SE-function mirrors the one proposed in \cite{aaai:FaberMW19}, we observe that
their different notions of strong equivalence and the ones proposed in this
paper coincide. Hence, Corollary~16 from \cite{aaai:FaberMW19} also holds for
our setting:

\begin{corollary}
  ELPs are strongly equivalent if and only if they have the same SE-function.
\end{corollary}

This is somewhat surprising, since, as we have seen, our notion of (ordinary)
equivalence is a more general one. From this observation, the following
complexity result immediately follows from Theorem~20 in \cite{aaai:FaberMW19}:

\begin{theorem}\label{thm:complexity}
  Checking whether two ELPs are strongly equivalent is \co\NP-complete.
\end{theorem}

\section{QSAT Solving: A Case Study}\label{sec:casestudies}

In this section, we will look at a brief case study to see how our new semantics
can be used to obtain intuitive information about problem solutions from an ELP
encoding for that problem in a much more straightforward way than with SE-CWVs.

To this end, we will be looking at the problem of solving quantified
satisfiability problems. In particular, we will look at the problem of solving
SAT formulas of the form $$\exists \varsX \forall \varsY \exists \varsZ \,
\Psi,$$ that is, existential QSAT formulas with three quantifier blocks. We will
further assume that the formula $\Psi$ is of the form $\bigwedge_{j=1}^n
(\ell_j^1 \vee \ell_j^2 \vee \ell_j^3)$, that is, a formula in conjunctive
normal form with three literals each. Each literal $\ell_j^i$ may be of the form
$w$ or $\neg \constant{w}$, where $\constant{w} \in \varsX \cup \varsY \cup
\varsZ$. We will follow the reduction from the proof of Theorem~5 in
\cite{ai:ShenE16}, and assume, w.l.o.g., that when all the variables in $\varsY$
are set to $\top$, then the formula becomes a tautology (independent of the
variables $\varsX$ and $\varsZ$). We will write our encoding in non-ground form,
since this should make it easier to read, and the intuitive meaning of the
non-ground rules is easier to grasp. Note that, to obtain the ground version of
the encoding, which then directly corresponds to the construction in
\cite{ai:ShenE16}, we simply substitute all the variables in the encoding with
all possible combinations of constants. This process is also called
\emph{grounding}; cf.\ e.g.\ \cite[Section 2.2.2]{ai:ShenE16}.

\paragraph{The Set of Facts.} Assume that the QSAT formula is represented by the
following facts:

\begin{itemize}
  \item $\fullatom{var_1}{\constant{x}}$, for each $\constant{x} \in \varsX$;
  \item $\fullatom{var_2}{\constant{y}}$, for each $\constant{y} \in \varsY$;
  \item $\fullatom{var_3}{\constant{z}}$, for each $\constant{z} \in \varsZ$;
    and
  \item $\fullatom{clause}{\constant{w}_1, \eta_1, \constant{w}_2, \eta_2,
    \constant{w}_3, \eta_3}$, for each clause $\ell_j^1 \vee \ell_j^2 \vee
    \ell_j^3$, $0 < j \leq n$, in $\Psi$, where for $0 < i \leq 3$, $\eta_j = 0$
    if $\ell_j^i = \constant{w}_i$, or $\eta_j = 1$ if $\ell_j^i = \neg
    \constant{w}_i$.
\end{itemize}

\paragraph{The Rules.} Now, let us construct the rules to actually solve this
problem. Note that the rules are non-ground, that is, they may contain
variables. Rules with variables can be seen as an abbreviation for multiple
copies of these rules, where each variable is replaced by some constant from the
atom domain. In this way, a non-ground ELP can be turned into a (ground) ELP as
defined in Section~\ref{sec:preliminaries}.

\begin{itemize}
  \item First, using epistemic negation, we guess an assignment for the
    variables in $\varsX$ of our formula: $$\fullatom{assign}{\varX, 0} \gets
    \eneg \fullatom{assign}{\varX, 1}, \fullatom{var_1}{\varX},$$
    $$\fullatom{assign}{\varX, 1} \gets \eneg \fullatom{assign}{\varX, 0},
    \fullatom{var_1}{\varX}.$$
  \item Then, we guess an assignment for the variables in $\varsY$:
    $$\fullatom{assign}{\varY, 0} \vee \fullatom{assign}{\varY, 1} \gets
    \fullatom{var_2}{\varY}.$$
  \item For the variables in $\varsX$, we use the standard ASP modelling
    technique of saturation \cite{amai:EiterG95}, in order to quantify over the
    variables in $\varsX$, as follows: $$\fullatom{assign}{\varZ, 0} \vee
    \fullatom{assign}{\varZ, 1} \gets \fullatom{var_3}{\varZ},$$
    $$\fullatom{assign}{\varZ, 0} \gets \relation{sat},$$
    $$\fullatom{assign}{\varZ, 1} \gets \relation{sat},$$ $$\relation{esat}
    \gets \eneg \relation{esat}, \eneg \neg \relation{sat}.$$
\end{itemize}

This completes the construction. Note that this non-ground encoding is a
straight-forward lifting of the ground version of this encoding presented in
\cite{ai:ShenE16}. Correctness hence follows from their paper. However, we would
like to point out the main difference between the SE-CWV semantics and the
semantics proposed in this paper. In the case of SE-CWV semantics, an ELP solver
would output a large set of answer sets, grouped together by their participation
in an SE-CWV. This makes it tedious to interpret the result, since this possibly
very large set of answer sets has to be carefully parsed and post-processed in
order to ascertain which atoms are always true, which are always false, and to
extract a truth assignment for the variables in $\varsX$ of the original QSAT
formula. While this process can be done without much difficulty by a computer,
the information is not directly accessible by a human user. In contrast, our
semantics would simply yield a number of CWVs, where, by construction, each CWV
would contain, for each variable $\constant{x} \in \varsX$, exactly one of the
two atoms $\fullatom{assign}{\constant{x}, 0}$ or
$\fullatom{assign}{\constant{x}, 1}$, allowing a hypothetical end-user to
directly extract a ``solution'' to the QSAT formula in the form of the truth
assignment to the variables in $\varsX$. In addition, every CWV under our
semantics represents exactly one valid truth assignment to the variables of
$\varsX$. In the authors' opinion, this seems like a much more intuitive
representation of the solutions of the QBF than the SE-CWVs.

\section{Conclusions}\label{sec:conclusions}

In this paper, we have presented a novel take on the semantics of epistemic
logic programs. When dealing with such ELPs under the semantics proposed by Shen
and Eiter \cite{ai:ShenE16}, we argued that the information that is actually of
interest to an end user is ``hidden'' in the so-called epistemic guess, and is
only implicitly present when SE-CWVs (which are sets of answer sets) are
computed. We therefore propose a novel notion of ``world view'' that directly
encodes this information as a three-valued assignment to the atoms of an ELP,
presenting immediately accessible information to the end user about which atoms
are always true, always false, or neither. We investigated complexity questions
and notions of equivalence between ELPs, as well as the relationship between our
new semantics and the one in \cite{ai:ShenE16}.

Future work will include implementing our semantics into an ELP solving system,
and investigating, whether practical solving shortcuts can be found, since, as
opposed to the SE-CWV semantics, not all the answer sets of the epistemic
reducts need to be computed. Further, note that our new definition of the world
view is not restricted to the semantics proposed in \cite{ai:ShenE16}, but can
also be directly applied to other semantics, e.g.\ the ones proposed by Gelfond
in \cite{aaai:Gelfond91,lpnmr:Gelfond11}, or by Kahl et al. in
\cite{diss:Kahl14,logcom:KahlWBGZ15}. We would like to do a similar
investigation w.r.t.\ these semantics, as the one done in this paper w.r.t.\ the
semantics from \cite{ai:ShenE16}.

%

\appendix

\bibliographystyle{eptcs}
\bibliography{references}

\end{document}